\documentclass[]{article}

\usepackage{amsfonts,amssymb,amsmath,amsthm}
\usepackage{stmaryrd}
\usepackage{nicefrac}
\usepackage{pgf}
\usepackage{tikz}
	\usetikzlibrary{trees,decorations,arrows,automata,positioning,plotmarks,
	shapes,backgrounds,petri}
\usepackage{xspace}
\usepackage{extarrows}
\usepackage{color}
\usepackage{paralist}
\usepackage[left=2.25cm,right=2.25cm,top=3cm,bottom=3cm]{geometry}
\usepackage{hyperref}

\usepackage{DistributabilityMobileAmbients}

\title{On the Distributability of Mobile Ambients (Technical Report)}

\author{Kirstin Peters
	\qquad\qquad Uwe Nestmann
	\vspace{0.5em}\\
	TU Berlin, Germany
}

\begin{document}

\maketitle

\begin{abstract}
	Modern society is dependent on distributed software systems and to verify them different modelling languages such as mobile ambients were developed.
	To analyse the quality of mobile ambients as a good foundational model for distributed computation, we analyse the level of synchronisation between distributed components that they can express.
	Therefore, we rely on earlier established synchronisation patterns.
	It turns out that mobile ambients are not fully distributed, because they can express enough synchronisation to express a synchronisation pattern called \patternM.
	However, they can express strictly less synchronisation than the standard \piCal.
	For this reason, we can show that there is no good and distributability-preserving encoding from the standard \piCal into mobile ambients and also no such encoding from mobile ambients into the \joinCal, \ie the expressive power of mobile ambients is in between these languages.
	Finally, we discuss how these results can be used to obtain a fully distributed variant of mobile ambients.

	This paper is an extended version of \cite{petersNestmann18}.
\end{abstract}

\section{Introduction}

Modern society is increasingly dependent on large-scale software systems that are distributed, collaborative, and communication-centred. Most of the existing approaches that analyse the distributability of concurrent systems use special formalisms often equipped with an explicit notion of location, \eg \cite{bestDarondeau11} in Petri nets or the distributed \piCal \cite{hennessy07}. Other approaches implement locations implicitly, as \eg the parallel operator in the \piCal that combines different distributed components of a system. In the latter case, we consider \emph{distributability} and, thus, all possible explicitly-located variants of a calculus.

The \piCal \cite{milnerParrowWalker92} is a well-known and frequently used process calculus to model concurrent systems. Therein, intuitively, the \emph{degree of distributability} cor\-res\-ponds to the number of parallel components that can act independently. Practical experience, though, has shown that it is not possible to implement every \piCal term---not even every asynchronous one---in an asynchronous setting while preserving its degree of distributability. To overcome these problems \eg the \joinCal \cite{levy97} or the distributed \piCal \cite{hennessy07} were introduced as models of distributed computation.

To analyse the quality of an approach as a good foundational model for distributed computation, we compare the expressiveness of different such models \wrt to their power to express synchronisation between distributed components.
Such synchronisations make the implementation of terms in an asynchronous setting difficult and, thus, indicate languages that are not suitable to describe distributed computation.
In particular, we try to identify hidden sources of synchronisation, \ie synchronisation that was not intended with the design of the calculus.

\paragraph{Distributability and Synchronisation Patterns.}
To analyse the degree of distribution in process calculi and to compare different calculi by their power to express synchronisation, \cite{petersNestmannGoltz13, peters12} defines a criterion for the preservation of distributability in encodings and introduces synchronisation patterns to describe minimal forms of synchronisation.
Process calculi are then separated by their power to express such synchronisation patterns and, thus, by the kinds of synchronisation that they contain.
Therefore, we show that no good and distributability-preserving encoding can exist from a calculus with enough synchronisation to express some synchronisation pattern into a calculus that cannot express this pattern.
In this sense, synchronisation patterns have two purposes:
\begin{inparaenum}[(1)]
	\item First, they describe some particular form or level of synchronisation in an abstract and model-independent way. Thereby, they help to spot forms of synchronisation---in particular, forms of synchronisation that were not intended with the design of the respective calculus.
	\item Second, they allow to separate calculi along their ability to express the respective pattern and the respective level of synchronisation.
\end{inparaenum}

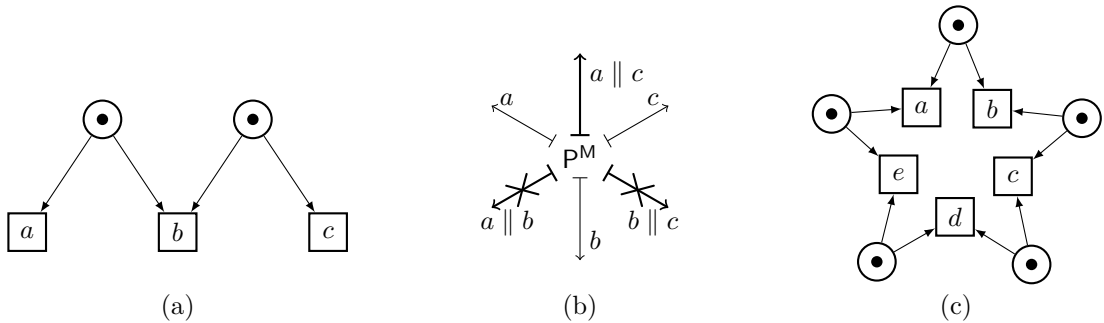
\begin{figure}[t]
	\centering
	\tikzstyle{place}=[circle,draw=black,thick,minimum size=5mm]
	\tikzstyle{transition}=[rectangle,draw=black,thick,minimum size=5mm]
	\begin{tikzpicture}
		\node[place,tokens=1]	(p) at (1, 1.5) {};
		\node[place,tokens=1]	(q) at (3, 1.5) {};
		\node[transition]		(a) at (0, 0) {$ a $};
		\node[transition]		(b) at (2, 0) {$ b $};
		\node[transition]		(c) at (4, 0) {$ c $};

		\draw[-latex] (p) -- (a);
		\draw[-latex] (p) -- (b);
		\draw[-latex] (q) -- (b);
		\draw[-latex] (q) -- (c);

		\node (l) at (2, -1) {(a)};
	\end{tikzpicture}
	\hspace*{4em}
	\begin{tikzpicture}
		\foreach \x/\xtext in {1/a,2/ab,3/b,4/bc,5/c,6/ac}
        {
            \path (360*\x/6+90:1.5) node (\xtext) {};
        }
		\node (p)						{$ \PM $};

		\path[|->]			(p) edge node[near end, above] {$ a $}	(a);
		\path[|->]			(p) edge node[near end, right] {$ b $}	(b);
		\path[|->]			(p) edge node[near end, above] {$ c $}	(c);
		\path[|->,thick]	(p) edge node[sloped, scale=2] {$ \times $} node[near end, below] {$ a \parallel b $}	(ab);
		\path[|->,thick]	(p) edge node[near end, right] {$ a \parallel c $}	(ac);
		\path[|->,thick]	(p) edge node[sloped, scale=2] {$ \times $} node[near end, below] {$ b \parallel c $}	(bc);

		\node (l) at (0, -2) {(b)};
	\end{tikzpicture}
	\hspace*{4em}
	\begin{tikzpicture}
		\foreach \x/\xtext in {1/e,2/d,3/c,4/b,5/a}
        {
            \path (360*\x/5+125:0.8) node[transition] (\xtext) {$\xtext$};
            \path (360*\x/5-55:1.75) node[place,tokens=1] (p\x) {};
        }

        \draw[-latex] (p2) -- (a);
        \draw[-latex] (p2) -- (b);

        \draw[-latex] (p1) -- (b);
        \draw[-latex] (p1) -- (c);

        \draw[-latex] (p5) -- (c);
        \draw[-latex] (p5) -- (d);

        \draw[-latex] (p4) -- (d);
        \draw[-latex] (p4) -- (e);

        \draw[-latex] (p3) -- (e);
        \draw[-latex] (p3) -- (a);

		\node (l) at (0, -2) {(c)};
	\end{tikzpicture}
	\vspace*{-1em}
	\caption{A fully reachable pure \patternM in Petri nets (a), the \patternM as state in a transition system (b), and the synchronisation pattern \patternGreatM in Petri nets (c).}
	\label{fig:MInPetri}
	\label{fig:greatMInPetri}
\end{figure}

In \cite{petersNestmannGoltz13}, two synchronisation patterns, the pattern \patternM and the pattern \patternGreatM, are highlighted.
An \patternM, as visualised in Figure~\ref{fig:MInPetri}~(a), describes a Petri net that consists of two parallel transitions ($ a $ and $ c $) and one transition ($ b $) that is in conflict with both of the former. In other words, it describes a situation where either two parts of the net can proceed independently or they synchronise to perform a single transition together. \cite{glabbeekGoltzSchicke08b,glabbeekGoltzSchicke12} states that a Petri net specification can be implemented in an asynchronous, fully distributed setting iff it does not contain a fully reachable pure \patternM. Accordingly, they denote such Petri nets as distributable. They also present a description of a fully reachable pure \patternM as conditions on a state $ \PM $ in a step transition system, as visualized in Figure~\ref{fig:MInPetri}~(b), which allows us to directly use this pattern to reason about process calculi.
Note that $ a $, $ b $, and $ c $ in Figure~\ref{fig:MInPetri}~(b) are not labels. They serve just to distinguish different steps. Moreover, $ x \parallel y $ refer to the parallel execution of $ x $ and $ y $, given a step semantics.
Hence, a process calculus is distributable iff it does not contain a non-local \patternM.
A \patternGreatM is a chain of conflicting and distributable steps as they occur in an \patternM that build a circle of odd length.
The Figure~\ref{fig:greatMInPetri}~(c) nicely illustrates this circle of \patternM. There is \eg one \patternM consisting of the transitions $ a $, $ b $, and $ c $ with their corresponding two places. Another \patternM is build by the transitions $ b $, $ c $, and $ d $ with their corresponding two places and so on.

These patterns are then used to locate various $ \pi $-like calculi within a hierarchy with respect to the level of synchronisation that can be expressed in these languages.
More precisely, \cite{petersNestmannGoltz13} shows that
\begin{inparaenum}[(1)]
\item the \joinCal is distributed, because it does not contain either of the two synchronisation patterns,
\item the asynchronous \piCal and its extension with separate choice can express the pattern \patternM but no pattern \patternGreatM, whereas the standard \piCal with mixed choice contains \patternM's and \patternGreatM's.
\end{inparaenum}

\paragraph{Mobile Ambients.}
In the current paper, we use the technique derived in \cite{petersNestmannGoltz13} to analyse the degree of distribution in mobile ambients.
Mobile ambients were introduced in \cite{cardelliGordon98, cardelliGordon00}.
Similar to the \joinCal, mobile ambients were designed as a calculus for distributed systems.
But, in contrast to the \joinCal, they do contain the pattern \patternM, as we show in the following.
Accordingly, mobile ambients are not fully distributed and their implementation in a fully distributed setting is difficult.
Fortunately, the little amount of synchronisation that is contained in mobile ambients is not enough to express the \patternGreatM.
Thus, mobile ambients are less synchronous than, \eg, the standard pi-calculus.
Moreover, the nature of the pattern \patternM that we find in mobile ambients tells us what kind of features lead to synchronisation in mobile ambients.
More precisely, we show that synchronisation in mobile ambients results from the so-called $ \maOpenAct $-actions and the fact that different ambients may share the same name.
This observation allows us to discuss ways to obtain a variant of mobile ambients that is free of hidden synchronisations and can, thus, be implemented easily in a distributed setting.

\paragraph{Overview.}
Section~\ref{sec:technicalPreliminaries} introduces process calculi (\S~\ref{sec:processCalculi}), mobile ambients (\S~\ref{sec:mobileAmbients}), encodings (\S~\ref{sec:encodings}), and synchronisation patterns together with some results of \cite{petersNestmannGoltz13} (\S~\ref{sec:distributability}) that are necessary for this paper.
In Section~\ref{sec:MinMA}, we show that mobile ambients can express enough synchronisation to contain pattern \patternM and that this implies that there is no good and distributability-preserving encoding from mobile ambients into the \joinCal.
Section~\ref{sec:conflictsInMA} analyses the nature of conflicts in mobile ambients that limits the forms of synchronisation they can express. It is shown that mobile ambients do not contain \patternGreatM-patterns; this separates them from the standard \piCal.
The observations on the nature of synchronisation in mobile ambients is then used in Section~\ref{sec:distributeMA} to discuss ways to obtain a distributed variant of mobile ambients.
We conclude with Section~\ref{sec:conclusions}.
This paper is an extended version of \cite{petersNestmann18}.

\section{Technical Preliminaries}
\label{sec:technicalPreliminaries}

We start with some general observations on process calculi and the relevant notions that we need for the comparison of process calculi as described in \cite{petersNestmannGoltz13}.
Then we describe the calculus of mobile ambients as introduced in \cite{cardelliGordon98, cardelliGordon00}
and ``good'' encodings as defined in \cite{gorla10}.
Finally, we shortly revise the results of \cite{petersNestmannGoltz13} that are relevant for our analysis of mobile ambients.

\subsection{Process Calculi}
\label{sec:processCalculi}

A \emph{process calculus} is a language $ \lang = \ProcCal{\proc}{\step} $ that consists of a set of process terms $ \proc $ (its syntax) and a relation $ {\step} : \proc \times \proc $ on process terms (its reduction semantics). We often refer to process terms also simply as processes or as terms and use upper case letters $ P, Q, R, \ldots, P', P_1, \ldots $ to range over them.

Assume a countably-infinite set $ \names $, whose elements are called \emph{names}. We use lower case letters such as $ a, b, c, \ldots, a', a_1, \ldots $ to range over names. Let $ \tau \notin \names $.
The \emph{syntax} of a process calculus is usually defined by a context-free grammar defining operators, \ie functions $ \operatorname{op} : \names^n \times \proc^m \to \proc $. An operator of arity $ 0 $, \ie $ m = 0 $, is a \emph{constant}. The arguments that are again process terms are called \emph{subterms} of $ P $.

\begin{definition}[Subterms]
	\label{def:subterms}
	Let $ \ProcCal{\proc}{\step} $ be a process calculus and $ P \in \proc $. The set of \emph{subterms} of $ P = \operatorname{op}\left( x_1, \ldots, x_n, P_1, \ldots, P_m \right) $ is defined recursively as $ \Set{ P } \cup \Set{ P' \mid \exists i \in \Set{ 1, \ldots, m } \logdot P' \text{ is a subterm of } P_i } $.
\end{definition}

\noindent
With Definition~\ref{def:subterms}, every term is a subterm of itself and constants have no further subterms.
We add the special constant $ \success $ to each process calculus. Its purpose is to denote \emph{success} (or \emph{successful termination}) which allows us to compare the abstract behaviour of terms in different process calculi as described in Section~\ref{sec:encodings}.
Therefore, we require that each language defines a predicate $ P\hasSuccess $ that holds if the term $ P $ is successful (or has terminated successfully). Usually, this predicate holds if $ P $ contains an occurrence of $ \success $ that is \emph{unguarded} (see mobile ambients below).

A \emph{scope} defines an area in which a particular name is known and can be used. For several reasons, it can be useful to restrict the scope of a name. For instance to forbid interaction between two processes or with an unknown and, hence, potentially untrusted environment. Names whose scope is restricted such that they cannot be used beyond their scope are called \emph{bound names}. The remaining names are called \emph{free names}.
As ususal, we define three sets of names occurring in a process term: the set $ \Names{P} $ of all of $ P $'s names, and its subsets $ \FreeNames{P} $ of free names and $ \BoundNames{P} $ of  bound names. In the case of bound names, their syntactical representation as lower case letters serves as a place holder for any fresh name, \ie any name that does not occur elsewhere in the term. To avoid confusion between free and bound names or different bound names, bound names can be replaced with fresh names by \emph{{\alphaConv}}. We write $ P \equiv_\alpha Q $ if $ P $ and $ Q $ differ only by {\alphaConv}.

We assume that the \emph{semantics} is given as an \emph{operational semantics} consisting of inference rules defined on the operators of the language \cite{plotkin04}. For many process calculi, the semantics is provided in two forms, as \emph{reduction semantics} and as \emph{labelled transition semantics}. We assume that at least the reduction semantics $ \step $ is given as part of the definition, because its treatment is easier in the context of encodings.
A single application of the reduction semantics is called a \emph{(reduction) step} and is written as $ P \step P' $. If $ P \step P' $, then $ P' $ is called \emph{derivative} of $ P $. Let $ P \step $ (or $ P \noStep $) denote the existence (absence) of a step from $ P $, and let $ \steps $ denote the reflexive and transitive closure of $ \step $. A sequence of reduction steps is called a \emph{reduction}.
We write $ P \step^{\omega} $ if $ P $ has an infinite sequence of steps and call $ P $ \emph{convergent} if $ \neg\left( P \step^{\omega} \right) $. We also use \emph{{execution}} to refer to a reduction starting from a particular term. A \emph{maximal execution} of a process $ P $ is a reduction starting from $ P $ that cannot be further extended, \ie that is either infinite or of the form $ P \steps P' \noStep $.

We extend the predicate $ P\hasSuccess $ to reachability of success. A term $ P \in \proc $ reaches success, written as $ P\reachSuccess $, if it reaches a derivative that is successful, \ie $ P\reachSuccess \deff \exists P' \logdot P \steps P' \wedge P'\hasSuccess $.
We write $ P\mustReachSuccessFinite $, if $ P $ reaches success in every finite maximal execution.

To reason about environments of terms, we use functions on process terms called contexts. More precisely, a \emph{context} $ \Context{}{}{\hole_1, \ldots, \hole_n} : \proc^n \to \proc $ with $ n $ holes is a function from $ n $ terms into a term, \ie given $ P_1, \ldots, P_n \in \proc $, the term $ \Context{}{}{P_1, \ldots, P_n} $ is the result of inserting $ P_1, \ldots, P_n $ in the corresponding order into the $ n $ holes of~$ \context $.

We assume the calculi \piMix for the standard \piCal (with mixed choice) as defined in \cite{milnerParrowWalker92} and its subcalculi the \piCal with only separate choice (\piSep), \ie there all parts of the same choice construct are either all guarded by an input or all guarded by an output prefix, and the asynchronous \piCal (\piAsyn) as introduced in \cite{boudol92, hondaTokoro91}.
Moreover, we assume the \joinCal (\join) as introduced in \cite{fournetGonthier96}.

\begin{definition}[Syntax, \cite{petersNestmannGoltz13}]
	\label{def:syntax}
	The sets of process terms are given by
	\begin{center}
		$ \begin{array}{ll}
			\procMix \deffTerms & P_1 \mid P_2 \sep \success \sep \Res{n}{P} \sep !P \sep \sum_{i \in \indexSet} \guard_i.P_i\\
				& \guard \deffTerms \Output{y}{z} \sep \Input{y}{x} \sep \tau \vspace*{0.5em}\\
			\procSep \deffTerms & P_1 \mid P_2 \sep \success \sep \Res{n}{P} \sep !P \sep \sum_{i \in \indexSet} \guard^O_i.P_i \sep \sum_{i \in \indexSet} \guard^I_i.P_i\\
				& \guard^O \deffTerms \Output{y}{z} \sep \tau \quad \text{ and } \quad
				\guard^I \deffTerms \Input{y}{x} \sep \tau \vspace*{0.5em}\\
			\procAsyn \deffTerms & \nullTerm \sep P_1 \mid P_2 \sep \success \sep \!\Res{n}{P} \sep !P \sep \Output{y}{z}\! \sep \Input{y}{x}.P \sep \!\tau.P \vspace*{0.5em}\\
			\procJoin \deffTerms & \nullTerm \sep P_1 \mid P_2 \sep \success \sep \JOutput{y}{z} \sep \JDefShort{D}{P}\\
				& J \deffTerms \JInput{y}{x} \sep J_1 \mid J_2 \quad \text{ and } \quad
				D \deffTerms J \triangleright P \sep D_1 \wedge D_2 
		\end{array} $
	\end{center}
	for some names $ n, x, y, z \in \names $ and a finite index set $ \indexSet $.
\end{definition}

In all languages the \emph{empty process} is denoted by $ \nullTerm $ and $ P_1 \mid P_2 $ defines \emph{parallel composition}. Within the \piCali \emph{restriction} $ \Res{n}{P} $ restricts the scope of the name $ n $ to the definition of $ P $ and $ !P $ denotes \emph{replication}. The process term $ \sum_{i \in \indexSet} \guard_i.P_i $ represents \emph{finite guarded choice}; as usual, the sum $ \sum_{i \in \Set{ 1, \ldots, n }} \guard_i.P_i $ is sometimes written as $ \guard_1.P_1 + \ldots + \guard_n.P_n $ and $ \nullTerm $ abbreviates the empty sum, \ie where $ \indexSet = \emptyset $. The input prefix $ \Input{y}{x} $ is used to describe the ability of receiving the value $ x $ over link $ y $ and, analogously, the output prefix $ \Output{y}{z} $ describes the ability to send a value $ z $ over link $ y $. The prefix $ \tau $ describes the ability to perform an internal, not observable action. The choice operators of \piMix and \piSep require that all branches of a choice are guarded by one of these prefixes. We omit the match prefix, because it does not influence the results.

In $ \procJoin $ the operator $ \JOutput{y}{z} $ describes an output prefix similar to $ \procAsyn $.
A \emph{definition} $ \JDefShort{D}{P} $ defines a new receiver on fresh names, where $ D $ consists of one or several elementary definitions $ J \triangleright P $ connected by $ \wedge $, $ J $ potentially joins several reception patterns $ \JInput{y}{x} $ connected by $ \mid $, and $ P $ is a process. Note that $ \JDefShort{D}{P} $ unifies the concepts of restriction, input prefix, and replication of the \piCal.

As usual, the continuation $ \nullTerm $ is often omitted, so \eg $ \Input{y}{x}\!.\nullTerm $ becomes $ \Input{y}{x} $.
In addition, for simplicity in the presentation of examples, we sometimes omit an action's object when it does not effectively contribute to the behaviour of a term, \eg $ \Input{y}{x}.\nullTerm $ is written as $ \In{y}.\nullTerm $ or just $ \In{y} $, and $ \JDef{\JInput{y}{x}}{\nullTerm}{\JOutput{y}{z}} $ is abbreviated as $ \JDef{\JIn{y}}{\nullTerm}{\JOut{y}} $.
Moreover, let $ \Res{\tilde{x}}{P} $ abbreviate the term $ \Res{x_1}{\ldots \Res{x_n}{P}} $.

The definitions of free and bound names are completely standard, \ie names are bound by restriction and as parameter of input and $ \Names{P} = \FreeNames{P} \cup \BoundNames{P} $ for all $ P $. In the \joinCal the definition $ \JDefShort{D}{P} $ binds for all elementary definitions $ J_i \triangleright P_i $ in $ D $ and all join pattern $ \JInput{y_{i, j}}{x_{i, j}} $ in $ J_i $ the \emph{received variables} $ x_{i, j} $ in the corresponding $ P_i $ and the \emph{defined variables} $ y_{i, j} $ in $ P $.

To compare process terms, process calculi usually come with different well-studied equivalence relations (see \cite{vanglabbeek01} 
for an overview). A special kind of equivalence with great importance to reason about processes are \emph{congruences}, \ie the closure of an equivalence with respect to contexts.
Process calculi usually come with a special congruence $ {\equiv} \mathrel{\subseteq} {\proc \times \proc} $ called \emph{structural congruence}. Its main purpose is to equate syntactically different process terms that model quasi-identical behaviour. For the above variants of the \piCal we have:
\begin{center}
	$ P \equiv Q \quad \text{ if } P \equiv_\alpha Q \hspace*{2em}
	P \mid \nullTerm \equiv P \hspace*{2em}
	P \mid Q \equiv Q \mid P \hspace*{2em}
	P \mid \left( Q \mid R \right) \equiv \left( P \mid Q \right) \mid R \hspace*{2em}
	!P \equiv P \mid !P $\\
	$ \Res{n}{\nullTerm} \equiv \nullTerm \hspace*{2em}
	\Res{n}{\Res{m}{P}} \equiv \Res{m}{\Res{n}{P}} \hspace*{2em}
	P \mid \Res{n}{Q} \equiv \Res{n}{\left( P \mid Q \right)} \quad \text{ if } n \notin \FreeNames{P} $
\end{center}
The entanglement of input prefix and restriction within the definition operator of the \joinCal limits the flexibility of relations defined by sets of equivalence equations. Instead structural congruence is given by an extension of the chemical approach in \cite{berryBoudol90} by the heating and cooling rules. They operate on so-called solutions $ \mathcal{R} \; \vdash \; \mathcal{M} $, where $ \mathcal{R} $ and $ \mathcal{M} $ are multisets. We have
\begin{inparaenum}[(1)]
	\item $ \vdash P \mid Q \rightleftharpoons \; \vdash P, Q $,
	\item $ D \wedge E \vdash \; \rightleftharpoons D, E \vdash $, and
	\item $ \vdash \JDefShort{D}{P} \rightleftharpoons \sigma_{dv}(D) \vdash \sigma_{dv}(P) $,
\end{inparaenum}
where only elements|separated by commas|that participate in the rule are mentioned and $ \sigma_{dv} $ instantiates the defined variables in $ D $ to distinct fresh names. Then $ P \equiv Q $ if $ P $ and $ Q $ differ only by applications of the $ \rightleftharpoons $-rules, \ie if $ \vdash P \rightleftharpoons \; \vdash Q $.

The semantics of the above variants of the \piCal is given by the axioms
\begin{center}
	$ \left( \ldots + \tau.P + \ldots \right) \step P \hspace*{1.4em} \left( \ldots + \Input{y}{x}.P + \ldots \right) \mid \left( \ldots + \Output{y}{z}.Q + \ldots \right) \step \Set{ \Subst{z}{x} }P \mid Q $
\end{center}
for \piMix and \piSep, the axioms $ \tau.P \step P $ and $ \Input{y}{x}.P \mid \Output{y}{z} \step \Set{ \Subst{z}{x} }P $ for \piAsyn, and the three rules
\begin{center}
	$ \dfrac{P \step P'}{P \mid Q \step P' \mid Q} \hspace*{2em} \dfrac{P \step P'}{\Res{n}{P} \step \Res{n}{P'}} \hspace*{2em} \dfrac{P \equiv Q \quad \quad Q \step Q' \quad Q' \equiv P'}{P \step P'} $
\end{center}
that hold for all three variants \piMix, \piSep, and \piAsyn.
The operational semantics of \join is given by the heating and cooling rules (see structural congruence) and the reduction rule $ J \triangleright P \vdash \sigma_{rv}(J) \step J \triangleright P \vdash \sigma_{rv}(P) $, where $ \sigma_{rv} $ substitutes the transmitted names for the distinct received variables.

Recursion or replication distinguishes itself from other operators by the fact that (one of) its subterms can be copied within rules of structural congruence in the \piCal or by reduction rules in the \joinCal while the operator itself is usually never removed during reductions. We call such operators and capabilities \emph{recurrent}.
We denote the parts of a term that are removed in reduction steps as \emph{capabilities}.

\subsection{Mobile Ambients}
\label{sec:mobileAmbients}

\emph{Mobile ambients} (\MA) were introduced in \cite{cardelliGordon98, cardelliGordon00} as a process calculus for distributed systems with mobile computations. They define \emph{ambients} as bounded places on that computations may happen and that can be moved (with their computations).
Their syntax is defined in two stages: the first stage describes \emph{ambient processes} and the nesting of ambients; the second stage describes the movements of ambients.

\begin{definition}[Syntax, \cite{cardelliGordon00}]
	The set of \emph{ambient processes} $ \procMA $ is given as
	\begin{center}
		$ \begin{array}{ll}
			\procMA \deffTerms & \nullTerm \sep P_1 \mid P_2 \sep \success \sep \maRes{n}{P} \sep \maRep{P} \sep \maLoc{n}{P} \sep M.P\\
				& M \deffTerms \maIn{n} \sep \maOut{n} \sep \maOpen{n}
		\end{array} $
	\end{center}
	for some names $ n \in \names $.
\end{definition}

The \emph{empty process} is denoted by $ \nullTerm $ and $ P_1 \mid P_2 $ define \emph{parallel composition}. \emph{Restriction} $ \maRes{n}{P} $ restricts the scope of the name $ n $ to the definition of $ P $.  \emph{Replication} $ !P $ provides potentially infinitely many copies of $ P $.
The $ \maLoc{n}{P} $ describes an \emph{ambient} $ n $ in which the process $ P $ is located. Ambients may exhibit a tree structure induced by the nesting of ambient brackets. The term $ M.P $ defines the exercise of capability $ M $, which could be either ``$ \maIn{n} $'' to \emph{enter} ambient $ n $, or ``$ \maOut{n} $'' to \emph{exit} from ambient $ n $, or ``$\maOpen{n} $'' to \emph{open} ambient~$ n $.
As usual, the continuation $ \nullTerm $ is often omitted.
Moreover, we often abbreviate $ \maLoc{n}{\nullTerm} $ by $ \maLoc{n}{} $
and let $ \maRes{\tilde{x}}{P} $ abbreviate the term $ \maRes{x_1}{\ldots \maRes{x_n}{P}} $.

Restriction is the only binder of mobile ambients, \ie the names are bound by restriction and all names of a process that are not bound by restriction are free.
The ``$ . $'' in $ M.P $ denotes sequential composition, where the $ M $ guards the subterm $ P $.
A subterm of a process is unguarded if it is not hidden behind a guard.
As usual, $ P\hasSuccess $ if $ P $ contains an unguarded occurrence of success.

For mobile ambients, \cite{cardelliGordon00} define structural congruence as the least congruence that satisfies the rules of $ \equiv $ defined above and additionally the rules $ !\nullTerm \equiv \nullTerm $ and $ \maRes{n}{\left( \maLoc{m}{P} \right)} \equiv \maLoc{m}{\maRes{n}{P}} $ if $ n \neq m $.

The reduction semantics of mobile ambients in \cite{cardelliGordon00} consists of the axioms
\begin{center}
	$ \maLoc{n}{\maIn{m}.P \mid Q} \mid \maLoc{m}{R} \step \maLoc{m}{\maLoc{n}{P \mid Q} \mid R} $ \vspace*{0.3em}\\
	$ \maLoc{m}{\maLoc{n}{\maOut{m}.P \mid Q} \mid R} \step \maLoc{n}{P \mid Q} \mid \maLoc{m}{R} \hspace*{2em} \maOpen{n}.P \mid \maLoc{n}{Q} \step P \mid Q $
\end{center}
and the rules:
\begin{center}
	$ \dfrac{P \step P'}{\maRes{n}{P} \step \maRes{n}{P'}} \hspace*{2em} \dfrac{P \step P'}{\maLoc{n}{P} \step \maLoc{n}{P'}} \hspace*{2em} \dfrac{P \step P'}{P \mid R \step P' \mid R} \hspace*{2em} \dfrac{P \equiv Q \quad \quad Q \step Q' \quad Q' \equiv P'}{P \step P'} $
\end{center}

The first axiom moves an ambient $ n $ with all its content (except for the consumed $ \maIn{m}{} $-capability) into a sibling ambient with name $ m $, where it is composed in parallel to the content of $ m $.
The second axiom allows an ambient $ n $ with all its content (except for the consumed $ \maOut{m}{} $-capability) to exit its parent ambient $ m $. As result ambient $ n $ is placed in parallel to $ m $.
The third axiom dissolves the boundary of an ambient named $ n $ that is located at the same level as the $ \maOpenAct $-capability.
The next three rules propagate reduction across scopes, ambient nesting, and parallel composition.
By the last rule reductions are defined modulo structural congruence.

Note that \cite{cardelliGordon00} explicitly states, that the same name can be used to name different ambients, \ie ambients with separate identities.
Moreover, if there are several ambients with the same name at the same hierarchical level all $ \maInAct $ and $ \maOpenAct $-capabilities that affect an ambient with this name can chose freely (non-deterministically) between the alternatives.

Following \cite{petersNestmannGoltz13}, we denote the operator $ !P $ for replication as recurrent, because (in contrast to the other operators) it is itself never removed during reductions.
Similarly, we denote an ambient that is not opened or moved in a step as recurrent for this step and, otherwise, as non-recurrent \wrt this step.
To distinguish between different occurrences of syntactically the same subterm in a term, we assume that all capabilities of processes in the following are implicitly labelled as described in \cite{petersNestmannGoltz13}.

\subsection{Encodings and Quality Criteria}
\label{sec:encodings}

Let $ \sourceLang = \ProcCal{\procSource}{\stepSource} $ and $ \targetLang = \ProcCal{\procTarget}{\stepTarget} $ be two process calculi, denoted as \emph{source} and \emph{target language}. An \emph{encoding} from $ \sourceLang $ into $ \targetLang $ is a function $ \arbitraryEncoding : \procSource \to \procTarget $. We often use $ S, S', S_1, \ldots $ to range over $ \procSource $ and $ T, T', T_1, \ldots $ to range over $ \procTarget $. Encodings often translate single source term steps into a sequence or pomset of target term steps. We call such a sequence or pomset an \emph{emulation} of the corresponding source term step.

To analyse the quality of encodings and to rule out trivial or meaningless encodings, they are augmented with a set of quality criteria. In order to provide a general framework, Gorla in \cite{gorla10} suggests five criteria well suited for language comparison. They are divided into two structural and three semantic criteria. The structural criteria include
\begin{inparaenum}[(1)]
	\item \emph{compositionality} and
	\item \emph{name invariance}. The semantic criteria include
	\item \emph{operational correspondence},
	\item \emph{divergence reflection}, and
	\item \emph{success sensitiveness}.
\end{inparaenum}
It turns out that we do not need the second criterion to derive the separation results of this paper. Thus, we omit it.
Note that a behavioural equivalence $ \asymp $ on the target language is assumed for the definition of name invariance and operational correspondence. 
Moreover, let $ \varphi : \names \to \names^k $ be a \emph{renaming policy}, \ie a mapping from a name to a vector of names that can be used by encodings to reserve special names, such that no two different names are translated into overlapping vectors of names.

Intuitively, an encoding is compositional if the translation of an operator is the same for all occurrences of that operator in a term. Hence, the translation of that operator can be captured by a context that is allowed in \cite{gorla10} to be parametrised on the free names of the respective source term.

\begin{definition}[Compositionality, \cite{gorla10}]
	\label{def:compositionality}
	The encoding $ \arbitraryEncoding $ is \emph{compositional} if, for every operator $ \mathbf{op} : \names^n \times \procSource^m \to \procSource $ of $ \sourceLang $ and for every subset of names $ N $, there exists a context $ \Context{N}{\mathbf{op}}{\hole_1, \ldots , \hole_{n + m}} : \names^n \times \procSource^m \to \procTarget $ such that, for all $ x_1, \ldots, x_n \in \names $ and all $ S_1, \ldots, S_m \in \procSource $ with $ \FreeNames{S_1} \cup \ldots \cup \FreeNames{S_m} = N $, it holds that $ \ArbitraryEncoding{\mathbf{op}\left( x_1, \ldots, x_n, S_1, \ldots, S_m \right)} = \Context{N}{\mathbf{op}}{\varphi\!\left( x_1 \right), \ldots, \varphi\!\left( x_n \right), \ArbitraryEncoding{S_1}, \ldots, \ArbitraryEncoding{S_m}} $.
\end{definition}

The first semantic criterion is operational correspondence. It consists of a soundness and a completeness condition. \emph{Completeness} requires that every computation of a source term can be emulated by its translation. \emph{Soundness} requires that every computation of a target term corresponds to some computation of the corresponding source term.

\begin{definition}[Operational Correspondence, \cite{gorla10}]
	\label{def:operationalCorrespondence}
	The encoding $ \arbitraryEncoding $ satisfies \emph{operational correspondence} if it satisfies:
	\begin{center}
		{\tabcolsep3pt
		\begin{tabular}{ll}
			\emph{Completeness}: & For all $ S \stepsSource S' $, it holds $ \ArbitraryEncoding{S} \stepsTarget \asymp \ArbitraryEncoding{S'} $.\\
			\emph{Soundness}: & For all $ \ArbitraryEncoding{S} \stepsTarget T $, there exists an $ S' $ such that $ S \stepsSource S' $ and $ T \stepsTarget \asymp \ArbitraryEncoding{S'} $.
		\end{tabular}}
	\end{center}
\end{definition}

\noindent
The definition of operational correspondence relies on the equivalence $ \asymp $ to get rid of junk possibly left over within computations of target terms. Sometimes, we refer to the completeness criterion of operational correspondence as \emph{operational completeness} and, accordingly, for the soundness criterion as \emph{operational soundness}.

The next criterion concerns the role of infinite computations in encodings.

\begin{definition}[Divergence Reflection, \cite{gorla10}]
	\label{def:divergenceReflection}
	The encoding $ \arbitraryEncoding $ \emph{reflects divergence} if, for every source term $ S $, $ \ArbitraryEncoding{S} \stepTarget^{\omega} $ implies $ S \stepSource^{\omega} $.
\end{definition}

The last criterion links the behaviour of source terms to the behaviour of their encodings. With Gorla \cite{gorla10}, we assume a \emph{success} operator $ \success $ as part of the syntax of both the source and the target language. Since $ \success $ cannot be further reduced and $ \Names{\success} = \FreeNames{\success} = \BoundNames{\success} = \emptyset $, the semantics and structural congruence of a process calculus are not affected by this additional constant operator.
We choose may-testing to test for the reachability of success, \ie $ P\reachSuccess \deff \exists P' \logdot P \steps P' \wedge P'\hasSuccess $. However, this choice is not crucial. An encoding preserves the abstract behaviour of the source term if it and its encoding answer the tests for success in exactly the same way.

\begin{definition}[Success Sensitiveness, \cite{gorla10}]
	\label{def:successSensitiveness}
	The encoding $ \arbitraryEncoding $ is \emph{success-sensitive} if, for every source term $ S $, $ S\reachSuccess $ iff $ \ArbitraryEncoding{S}\reachSuccess $.
\end{definition}

\noindent
This criterion only links the behaviours of source terms and their literal translations, but not of their derivatives. To do so, Gorla relates success sensitiveness and operational correspondence by requiring that the equivalence on the target language never relates two processes with different success behaviours.

\begin{definition}[Success Respecting, \cite{gorla10}]
	\label{def:asympSuccessRespecting}
	$ \asymp $ is \emph{success respecting} if, for every $ P $ and $ Q $ with $ P\reachSuccess $ and $ Q\notReachSuccess $, it holds that $ P \not\asymp Q $.
\end{definition}

\noindent
By \cite{gorla10} a ``good'' equivalence $ \asymp $ is often defined in the form of a barbed equivalence (as described e.g. in \cite{milnerSangiorgi92}) or can be derived directly from the reduction semantics and is often a congruence, at least with respect to parallel composition. For the separation results presented in this paper, we require only that $ \asymp $ is a success respecting reduction bisimulation.

\begin{definition}[(Weak) Reduction Bisimulation]
	\label{def:asympReductionBisim}
	The equivalence $ \asymp $ is a \emph{(weak) reduction bisimulation} if, for every $ T_1, T_2 \in \procTarget $ such that $ T_1 \asymp T_2 $, for all $ T_1 \stepsTarget T_1' $ there exists a $ T_2' $ such that $ T_2 \stepsTarget T_2' $ and $ T_1' \asymp T_2' $.
\end{definition}

Note that the best known encoding from the asynchronous \piCal into the \joinCal in \cite{fournetGonthier96} is not compositional, but consists of an inner, compositional encoding surrounded by a fixed context|the implementation of so-called firewalls|that is parametrised on the free names of the source term. In order to capture this and similar encodings and as done in \cite{petersNestmannGoltz13} we relax the definition of compositionality in our notion of a good encoding.

\begin{definition}[Good Encoding]
	\label{def:goodEncoding}
	We consider an encoding $ \arbitraryEncoding $ to be \emph{good} if it is
	\begin{inparaenum}[(1)]
		\item either compositional or consists of an inner, compositional encoding surrounded by a fixed context that can be parametrised on the free names of the source term,
		\item satisfies operational correspondence,
		\item reflects divergence, and
		\item is success-sensitive.
	\end{inparaenum}
	Moreover we require that the equivalence $ \asymp $ is a success respecting (weak) reduction bisimulation.
\end{definition}

In this case a good encoding respects also the ability to reach success in all finite maximal executions.

\begin{lemma}[\cite{petersNestmannGoltz13TR}]
	\label{lem:mustSuccessRespecting}
	For all success respecting reduction bisimulations $ \asymp $ and all convergent target terms $ T_1, T_2 $ such that $ T_1 \asymp T_2 $, it holds $ T_1\mustReachSuccessFinite $ iff $ T_2\mustReachSuccessFinite $.
\end{lemma}

Then success sensitiveness preserves the ability to reach success in all finite maximal executions.

\begin{lemma}[\cite{petersNestmannGoltz13TR}]
	\label{lem:mustSuccessSensitiveness}
	For all operationally sound, divergence reflecting, and success-sensitive encodings $ \arbitraryEncoding $ with respect to some success respecting equivalence $ \asymp $ and for all convergent source terms $ S $, if $ S\mustReachSuccessFinite $ then $ \ArbitraryEncoding{S}\mustReachSuccessFinite $.
\end{lemma}

\subsection{Distributability and Synchronisation Pattern}
\label{sec:distributability}

Intuitively, a distribution of a process means the extraction (or: separation) of its (sequential) components and their association to different locations.
However, not all process calculi in the literature---as \eg the standard \piCal in \cite{milnerParrowWalker92}---consider locations explicitly. For the calculi without an explicit notion of location \cite{petersNestmannGoltz13} defines a general notion of \emph{distributability} that focuses on the possible division of a process term into components.
Accordingly, a process $ P $ is distributable into $ P_1 , \ldots , P_n $, if we find some distribution that extracts $ P_1 , \ldots , P_n $ from within $ P $ onto different locations.

\begin{definition}[Distributability, \cite{petersNestmannGoltz13}]
	\label{def:degreeOfDistributability}
	Let $ \ProcCal{\proc}{\step} $ be a process calculus, $ \equiv $ be its structural congruence, and $ P \in \proc $. $ P $ is \emph{distributable} into $ P_1, \ldots, P_n \in \proc $ if there exists $ P' \equiv P $ such that
	\begin{compactenum}
		\item for all $ 1 \leq i \leq n $, $ P_i $ contains at least one capability or constant different from $ \nullTerm $ and $ P_i $ is an unguarded subterm of $ P' $ or, in case $ \equiv $ is given by a chemical approach, $ \vdash P' \rightleftharpoons \mathcal{R} \vdash P_i, \mathcal{M} $ for some multisets $ \mathcal{R}, \mathcal{M} $,
		\item in $ P_1, \ldots, P_n $ there are no two occurrences of the same capability, \ie no label occurs twice, and
		\item each guarded subterm and each constant (different from $ \nullTerm $) of $ P' $ is a subterm of at least one of the terms $ P_1, \ldots, P_n $.
	\end{compactenum}
	The \emph{degree of distributability} of $ P $ is the maximal number of distributable subterms of $ P $.
\end{definition}

\noindent
Accordingly, a pi-term $ P $ is distributable into $ P_1, \ldots, P_n $ if $ P \equiv \Res{\tilde{a}}{\left( P_1 \mid \ldots \mid P_n \right)} $. The $ \procJoin $-term $ \JDef{\JIn{a}\,}{\nullTerm}{\left( \JDef{\JIn{b}\,}{\JOutput{c}{a}}{\left( \JOut{a} \mid \JOut{b} \right)} \right)} $ is distributable into $ \JDef{\JIn{a}\,}{\nullTerm}{\JOut{a}} $ and $ \JDef{\JIn{b}\,}{\JOutput{c}{a}}{\JOut{b}} $, but \eg also into $ \JDef{\JIn{a}\,}{\nullTerm}{\nullTerm} $, $ \JDef{\JIn{b}\,}{\JOutput{c}{a}}{\nullTerm} $, $ \JOut{a} $, and $ \JOut{b} $, because
$ \vdash \JDef{\JIn{a}\,}{\nullTerm}{\left( \JDef{\JIn{b}\,}{\JOutput{c}{a}}{\left( \JOut{a} \mid \JOut{b} \right)} \right)}
\rightleftharpoons \JDefShort{\JIn{a}\,}{\nullTerm}, \JDefShort{\JIn{b}\,}{\JOutput{c}{a}} \vdash \JOut{a} \mid \JOut{b}
\rightleftharpoons \JDefShort{\JIn{a}\,}{\nullTerm}, \JDefShort{\JIn{b}\,}{\JOutput{c}{a}} \vdash \JOut{a}, \JOut{b}
\rightleftharpoons \; \vdash \JDef{\JIn{a}\,}{\nullTerm}{\nullTerm}, \JDef{\JIn{b}\,}{\JOutput{c}{a}}{\nullTerm}, \JOut{a}, \JOut{b} $.

Mobile ambients come with an explicit notion of locations: ambients.
A term of $ \procMA $ is distributable into pairwise intersected subsets of its outermost ambients.
Applying the Definition~\ref{def:degreeOfDistributability} results into exactly these distributable components.
Because of the rule $ !P \equiv P \mid !P $, the replication of an ambient, \eg by $ !(\maLoc{n}{P}) $ or $ !(\maRes{n}{\maLoc{n}{P}}) $, is a distributable recurrent operation.

Preservation of distributability  means that the target term is at least as distributable as the source~term.

\begin{definition}[Preservation of Distributability, \cite{petersNestmannGoltz13}]
	\label{def:distributabilityPreservation}
	An encoding $ \arbitraryEncoding : \procSource \to \procTarget $ \emph{preserves distributability} if for every $ S \in \procSource $ and for all terms $ S_1, \ldots, S_n \in \procSource $ that are distributable within $ S $ there are some $ T_1, \ldots, T_n \in \procTarget $ that are distributable within $ \ArbitraryEncoding{S} $ such that $ T_i \asymp \ArbitraryEncoding{S_i} $ for all $ 1 \leq i \leq n $.
\end{definition}

In essence, this requirement is a distributability-enhanced adaptation of operational completeness.
It respects both the intuition on distribution as separation on different locations---an encoded source term is at least as distributable as the source term itself---as well as the intuition on distribution as independence of processes and their executions---implemented by $ T_i \asymp \ArbitraryEncoding{S_i} $.

If a single process---of an arbitrary process calculus---can perform two different steps, \ie steps on capabilities with different labels, then we call these steps alternative to each other. Two alternative steps can either be in conflict or not; in the latter case, it is possible to perform both of them in parallel, according to some assumed step semantics.

\begin{definition}[Distributable Steps, \cite{petersNestmannGoltz13}]
	\label{def:distributableSteps}
	Let $ \ProcCal{\proc}{\step} $ be a process calculus and $ P \in \proc $ a process. Two alternative steps of $ P $ are in \emph{conflict}, if performing one step disables the other step, \ie if both reduce the same not recurrent capability. Otherwise they are \emph{parallel}.
	Two parallel steps of $ P $ are \emph{distributable}, if each recurrent capability reduced by both steps is distributable, else the steps are \emph{local}.
\end{definition}

\noindent
Remember that the ``same'' means ``with the same label'', \ie in $ \left( \maOpen{n} \mid \maLoc{n}{P_1} \mid \maLoc{n}{P_2} \right) $ the two steps that open one of the ambients $ n $ are in conflict but $ \left( \maOpen{n} \mid \maLoc{n}{P_1} \mid \maOpen{n} \mid \maLoc{n}{P_2} \right) $ can perform two parallel steps---using different $ \maOpen $-capabilities and ambients---to open both ambients $ n $.

Next we define parallel and distributable sequences of steps.

\begin{definition}[Distributable Executions, \cite{petersNestmannGoltz13}]
	\label{def:distributableSequences}
	Let $ \ProcCal{\proc}{\step} $ be a process calculus, $ P \in \proc $, and let $ A $ and $ B $ denote two executions of $ P $. $ A $ and $ B $ are in \emph{conflict}, if a step of $ A $ and a step of $ B $ are in conflict, else $ A $ and $ B $ are \emph{parallel}.
	Two parallel sequences of steps $ A $ and $ B $ are \emph{distributable}, if each pair of a step of $ A $ and a step of $ B $ is distributable.
\end{definition}

Two executions of a term $ P $ are distributable iff $ P $ is distributable into two subterms such that each performs one of these executions. Hence, an operationally complete encoding is distributability-preserving only if it preserves the distributability of sequences of source term steps.

\begin{lemma}[Distributability-Preservation, \cite{petersNestmannGoltz13}]
	\label{lem:distributabilityPreservation}
	An operationally complete encoding $ \arbitraryEncoding : \procSource \to \procTarget $ that preserves distributability also preserves distributability of executions, \ie for all source terms $ S \in \procSource $ and all sets of pairwise distributable executions of $ S $, there exists an emulation of each execution in this set such that all these emulations are pairwise distributable in $ \ArbitraryEncoding{S} $.
\end{lemma}

As described in the introduction, we consider a process calculus is distributable iff it does not contain a non-local \patternM.

\begin{definition}[Synchronisation Pattern \patternM, \cite{petersNestmannGoltz13}]
	\label{def:synchronisationPatternM}
	Let $ \ProcCal{\proc}{\step} $ be a process calculus and $ \PM \in \proc $ such that:
	\begin{compactenum}
		\item $ \PM $ can perform at least three alternative steps $ a\!\!: \PM \step P_a $, $ b\!\!: \PM \step P_b $, and $ c\!: \PM \step P_c $ such that $ P_a $, $ P_b $, and $ P_c $ are pairwise different.
		\item The steps $ a $ and $ c $ are parallel in $ \PM $.
		\item But $ b $ is in conflict with both $ a $ and $ c $.
	\end{compactenum}
	In this case, we denote the process $ \PM $ as \patternM.
	If the steps $ a $ and $ c $ are distributable in $ \PM $, then we call the \patternM \emph{non-local}. Otherwise, the \patternM is called \emph{local}.
\end{definition}

As shown in \cite{petersNestmannGoltz13}, all \patternM in the \joinCal (\join) are local but the asynchronous \piCal (\piAsyn) contains the non-local \patternM: $ \Output{y}{u} \mid \Input{y}{x}.P_1 \mid \Output{y}{v} \mid \Input{y}{x}.P_2 $ with $ P_1, P_2 \in \procAsyn $, where the steps $ a $, $ b $, and $ c $ are the reduction of the first out- and input, the first input and the second output, and the second out- and input, respectively.
Because of that, there is no good and distributability-preserving encoding from \piAsyn into \join.
To further distinguish different variants of the \piCal, \cite{petersNestmannGoltz13} introduces a second synchronisation pattern called \patternGreatM.
Interestingly, it reflects a well-known standard problem in the area of distributed systems, namely the problem of the dining philosophers \cite{dijkstra71}.

\begin{definition}[Synchronisation Pattern \patternGreatM, \cite{petersNestmannGoltz13}]
	\label{def:synchronisationPatternGreatM}
	Let $ \ProcCal{\proc}{\step} $ be a process calculus and $ \PS \in \proc $ such that:
	\begin{compactenum}
		\item $ \PS $ can perform at least five alternative reduction steps $ i : \PS \step P_i $ for $ i \in \Set{ a, b, c, d, e } $ such that the $ P_i $ are pairwise different.
		\item The steps $ a $, $ b $, $ c $, $ d $, and $ e $ form a circle such that $ a $ is in conflict with $ b $, $ b $ is in conflict with $ c $, $ c $ is in conflict with $ d $, $ d $ is in conflict with $ e $, and $ e $ is in conflict with $ a $. Finally,
		\item every pair of steps in $ \Set{ a, b, c, d, e } $ that is not in conflict due to the previous condition is parallel in $ \PS $.
	\end{compactenum}
	In this case, we denote the process $ \PS $ as \patternGreatM. The synchronisation pattern \patternGreatM is visualised by the Petri net in Figure~\ref{fig:greatMInPetri}~(c). If all pairs of parallel steps in $ \Set{ a, b, c, d, e } $ are distributable in $ \PS $, then we call the \patternGreatM \emph{non-local}. Otherwise, it is called \emph{local}.
\end{definition}

\noindent
Note that we need at least four steps in this cycle, to have two steps that are distributable, and a cycle of odd degree to distinguish different variants of the \piCal. Accordingly, the \patternGreatM is the smallest structure with these requirements.
To see the connection with the dining philosophers problem, consider the places in Figure~\ref{fig:greatMInPetri}~(c) as the chopsticks of the philosophers, \ie as resources, and the transitions as eating operations, \ie as steps consuming resources. Each step needs mutually exclusive access to two resources and each resource is shared among two subprocesses. If both resources are allocated simultaneously, eventually exactly two steps are performed.

\cite{petersNestmannGoltz13} then shows that the asynchronous \piCal (\piAsyn) and also the \piCal with separate choice (\piSep) do not contain the pattern \patternGreatM, whereas the standard \piCal (\piMix) with mixed choice has
\patternGreatM.

\begin{example}[Non-Local \patternGreatM in \piMix]
	\label{exa:patternGreatMInPiMix}
	Consider a term $ \piS{\piSE{1}, \ldots, \piSE{5}} \in \procMix $ such that
	\begin{align*}
		\piS{\piSE{1}, \ldots, \piSE{5}} = \Out{a} + b.\piSE{1} \mid \Out{b} + c.\piSE{2} \mid \Out{c} + d.\piSE{3} \mid \Out{d} + e.\piSE{4} \mid \Out{e} + a.\piSE{5}
	\end{align*}
	for some $ \piSE{1}, \ldots, \piSE{5} \in \Set{ \nullTerm, \success } $. Then, $ \piS{\piSE{1}, \ldots, \piSE{5}} $ can perform the steps $ a $, \ldots, $ e $, where the step $ i \in \Set{ a, \ldots, e } $ is a communication on channel $ i $.
	By Definition~\ref{def:synchronisationPatternGreatM}, $ \piS{\piSE{1}, \ldots, \piSE{5}} $ is a non-local \patternGreatM.
\end{example}

Actually, the above term $ \piS{\piSE{1}, \ldots, \piSE{5}} $ is a \patternGreatM in CCS with mixed choice, because for this counterexample the communication of values was not relevant. Adding (unused) values to the communication prefixes is straight forward.
By using the \patternGreatM-pattern $ \piS{\piSE{1}, \ldots, \piSE{5}} $ as counterexample, \cite{petersNestmannGoltz13, petersNestmannGoltz13TR} shows that there is no good and distributability-preserving encoding from \piMix into \piSep (or \piAsyn).

\section{Mobile Ambients are not Distributable}
\label{sec:MinMA}

Similar to the \joinCal, mobile ambients were designed in order to be distributed (or distributable), where ambients were introduced as an explicit representation of locations. But in opposite to the \joinCal there are non-local \patternM in mobile ambients, \ie some form of synchronisation between ambients.
\begin{example}[Non-Local \patternM in Mobile Ambients.]
	\label{exa:nonLocalMMA}
	Consider the $ \procMA $-term
	\begin{align*}
		& \maM = \left( \maOpen{n_1} \mid \maLoc{n_1}{P_1} \right) \mid \left( \maLoc{n_1}{\maIn{n_2}.P_2} \mid \maLoc{n_2}{P_3} \right)
	\end{align*}
	with $ P_1, P_2, P_3 \in \procMA $. $ \maM $ can perform modulo structural congruence the steps
	\begin{itemize}
		\item $ a $: $ \maM \step P_1 \mid \left( \maLoc{n_1}{\maIn{n_2}.P_2} \mid \maLoc{n_2}{P_3} \right) $
		\item $ b $: $ \maM \step \maLoc{n_1}{P_1} \mid \maIn{n_2}.P_2 \mid \maLoc{n_2}{P_3} $
		\item $ c $: $ \maM \step \left( \maOpen{n_1} \mid \maLoc{n_1}{P_1} \right) \mid \left( \maLoc{n_2}{\maLoc{n_1}{P_2} \mid P_3} \right) $
	\end{itemize}
	Here, the steps $ a $ and $ b $ compete for the non-recurrent $ \maOpenAct $-capability. The steps $ b $ and $ c $ compete for the right ambient $ n_1 $ that is non-recurrent in both steps. Hence, both of these pairs of steps are in conflict, while the pair of steps $ a $ and $ c $ is distributable. Thus $ \maM $ is a non-local \patternM.
\end{example}

Similar to the proof, that there is no good and distributability-preserving encoding from \piAsyn into \join, we use this $ \maM $ as a counterexample to show that there is no good and distributability-preserving encoding from \MA into \join.
Therefore, we instantiate the processes $ P_1 $, $ P_2 $, and $ P_3 $ such that the conflicting step $ b $ can be distinguished by success from the distributable steps $ a $ and $ c $.
We choose $ P_1 = \maLoc{n_3}{} $, $ P_2 = \maIn{n_3}.\success $, and $ P_3 = \maOpen{n_1} $, such that $ \maM $ reaches success iff the steps $ a $ and $ c $ are performed.
\begin{example}[Counterexample]
	The non-local \patternM
	\begin{align*}
		\maMS = \left( \maOpen{n_1} \mid \maLoc{n_1}{\maLoc{n_3}{}} \right) \mid \left( \maLoc{n_1}{\maIn{n_2}.\maIn{n_3}.\success} \mid \maLoc{n_2}{\maOpen{n_1}} \right)
	\end{align*}
	reaches success iff $ \maMS $ performs both of the distributable steps $ a $ and $ c $, where
	\begin{itemize}
		\item $ a $: $ \maMS \step S_a $ with $ S_a = \maLoc{n_3}{} \mid \left( \maLoc{n_1}{\maIn{n_2}.\maIn{n_3}.\success} \mid \maLoc{n_2}{\maOpen{n_1}} \right) $ and $ S_a\mustReachSuccessFinite $,
		\item $ b $: $ \maMS \step S_b $ with $ S_b = \maLoc{n_1}{\maLoc{n_3}{}} \mid \maIn{n_2}.\maIn{n_3}.\success \mid \maLoc{n_2}{\maOpen{n_1}} $ and $ S_b\notReachSuccess $, and
		\item $ c $: $ \maMS \step S_c $ with $ S_c = \left( \maOpen{n_1} \mid \maLoc{n_1}{\maLoc{n_3}{}} \right) \mid \maLoc{n_2}{\maLoc{n_1}{\maIn{n_3}.\success} \mid \maOpen{n_1}} $ and $ S_c\mustReachSuccessFinite $.
	\end{itemize}
\end{example}

Any good encoding that preserves distributability has to translate $ \maMS $ such that the emulations of the steps $ a $ and $ c $ are again distributable. However, the encoding can translate these two steps into sequences of steps, which allows to emulate the conflicts with the emulation of $ b $ by two different distributable steps. We show that every distributability-preserving encoding has to distribute $ b $ and, afterwards, that this distribution of $ b $ violates the criteria of a good encoding.

\begin{lemma}
	\label{lem:splitUpB}
	Every encoding $ \arbitraryEncoding : \procMA \to \procJoin $ that is good and distributability-preserving has to split up the conflict in $ \maMS $ of $ b $ with $ a $ and $ c $ such that there exists a maximal execution in $ \ArbitraryEncoding{\maMS} $ in which $ a $ is emulated but not $ c $, and vice versa.
\end{lemma}

In \cite{petersNestmannGoltz13} we show a similar result for all encodings from \piAsyn into \join (Lemma~4 in \cite{petersNestmannGoltz13}) using a counterexample E1.
Since the counterexample $ \maMS $ in \MA is in its properties very similar to the counterexample E1 of \cite{petersNestmannGoltz13}, the proof of Lemma~\ref{lem:splitUpB} is exactly the same as the proof of Lemma~4 in \cite{petersNestmannGoltz13} as presented in \cite{petersNestmannGoltz13TR}.
The main idea of this proof is as follows:
Any good encoding that preserves distributability has to translate $ \maMS $ such that the emulations of the steps $ a $ and $ c $ are again distributable. Moreover any good encoding has to translate the conflicts between $ a $ and $ b $ as well as between $ b $ and $ c $ into conflicts between the respective emulations. This either leads to a non-local \patternM again or it results into an emulation of $ b $ with at least two steps such that the conflicts with the emulation of $ b $ are emulated by two different steps. Next we show that this distribution of the conflict violates the criteria of a good encoding with respect to the considered source language, \ie \wrt our counterexample $ \maMS $ and an adaptation of this example.

Also the proof that there is no good and distributability-preserving encoding from \MA into \join is very similar to the proof for the non-existence of such an encoding from \piAsyn into \join in \cite{petersNestmannGoltz13, petersNestmannGoltz13TR}.

\begin{theorem}
	\label{thm:noGoodEncodingMA}
	There is no good and dis\-tri\-bu\-ta\-bi\-li\-ty-pre\-ser\-ving encoding from \MA into \join.
\end{theorem}

\begin{proof}
	Assume the opposite.
	Then there is a good and distributability-preserving encoding of $ \maMS $.
	By the proof of Lemma~\ref{lem:splitUpB}, there is a maximal execution of $ \ArbitraryEncoding{\maMS} $ in that $ a $ but not $ c $ is emulated or vice versa. Since $ S_a\mustReachSuccessFinite $ and $ S_c\mustReachSuccessFinite $ and because of success sensitiveness, the corresponding emulation leads to success. So there is an execution such that the emulation of $ a $ leads to success without the emulation of $ c $ or vice versa. Let us assume that $ a $ but not $ c $ is emulated. The other case is similar.

	For encodings with respect to the relaxed definition of compositionality in Definition~\ref{def:goodEncoding}, there exists a context $ \Context{\hole_1, \hole_2} : \procJoin^2 \to \procJoin $---the combination of the surrounding context and the context introduced by compositionality (Definition~\ref{def:compositionality})---such that $ \ArbitraryEncoding{\maMS} = \Context{}{}{\ArbitraryEncoding{S_1}, \ArbitraryEncoding{S_2}} $, where $ S_1, S_2 \in \procMA $ with $ S_1 = \maOpen{n_1} \mid \maLoc{n_1}{\maLoc{n_3}{}} $ and $ S_2 = \maLoc{n_1}{\maIn{n_2}.\maIn{n_3}.\success} \mid \maLoc{n_2}{\maOpen{n_1}} $.
	Let $ S_2' = \maLoc{n_1}{\maOut{n_2}.\maIn{n_3}.\success} \mid \maLoc{n_2}{\maOpen{n_1}} $.
	Since $ \FreeNames{S_2} = \FreeNames{S_2'} $, also $ S_1 \mid S_2' $ has to be translated by the same context, \ie $ \ArbitraryEncoding{S_1 \mid S_2'} = \Context{}{}{\ArbitraryEncoding{S_1}, \ArbitraryEncoding{S_2'}} $.
	$ \maMS $ and $ S_1 \mid S_2' $ differ only by a capability necessary for step $ c $, but step $ a $ and $ b $ are still possible.
	We conclude, that if $ \Context{}{}{\ArbitraryEncoding{S_1}, \ArbitraryEncoding{S_2}} $ reaches some $ T_a\mustReachSuccessFinite $ without the emulation of $ c $, then $ \Context{}{}{\ArbitraryEncoding{S_1}, \ArbitraryEncoding{S_2'}} $ reaches at least some state $ T_a' $ such that $ T_a'\reachSuccess $. Hence, $ \ArbitraryEncoding{S_1 \mid S_2'}\reachSuccess $ but $ \left( S_1 \mid S_2' \right)\notReachSuccess $ which contradicts success sensitiveness.
\end{proof}

\noindent
Note that the only differences in the proof above and the proof for the the non-existence of a good and distributability-preserving encoding from \piAsyn into \join in \cite{petersNestmannGoltz13TR} are the due to the different counterexample and the corresponding choice of its adaptation with $ S_2' $.

\section{Conflicts in Mobile Ambients}
\label{sec:conflictsInMA}

Both of the above-defined synchronisation patterns rely on the notion of conflict.
In mobile ambients, the same ambient can be considered as recurrent in one step, but non-recurrent in another step. This fact, \ie the existence of operators that are recurrent for some but non-recurrent for other steps, distinguishes mobile ambients from all other calculi considered in \cite{petersNestmannGoltz13} and generates a new notion of conflict.
\begin{example}[Asymmetric Conflict]
	Consider the mobile ambient term:
	\begin{align*}
		P = \maLoc{n_1}{\maIn{n_2}} \mid \maLoc{n_2}{\maIn{n_3}} \mid \maLoc{n_3}{}
	\end{align*}
	$ P $ can perform two alternative steps
	\begin{itemize}
		\item $ s_1: P \step P_1 $ with $ P_1 = \maLoc{n_2}{\maLoc{n_1}{} \mid \maIn{n_3}} \mid \maLoc{n_3}{} $ and
		\item $ s_2: P \step P_2 $ with $ P_2 = \maLoc{n_1}{\maIn{n_2}} \mid \maLoc{n_3}{\maLoc{n_2}{}} $
	\end{itemize}
	that both use the ambient $ n_2 $ (but no other operator is used in both steps). In $ s_1 $, the ambient $ n_2 $ is a recurrent capability but in $ s_2 $ the ambient $ n_2 $ is moved and, thus, is non-recurrent. Accordingly, $ s_2 $ disables $ s_1 $, \ie $ P_2 \noStep $, but not vice versa, \ie $ P_1 $ can perform the step $ s_2 $ such that $ P_1 \step \maLoc{n_3}{\maLoc{n_2}{\maLoc{n_1}{}}} $.
\end{example}
Accordingly, we denote a conflict as \emph{symmetric} if the steps compete for an operator that is non-recurrent in both, \ie if both steps disable the respective other step, and otherwise as \emph{asymmetric}.
The example above can be extended to a cyclic structure of odd degree. The term
\begin{align*}
	\maLoc{a}{\maIn{b}} \mid \maLoc{b}{\maIn{c}} \mid \maLoc{c}{\maIn{d}} \mid \maLoc{d}{\maIn{e}} \mid \maLoc{e}{\maIn{a}}
\end{align*}
even satisfies Definition~\ref{def:synchronisationPatternGreatM}, \ie it describes a non-local \patternGreatM, if we were to relax in the required conflicts in Definition~\ref{def:synchronisationPatternGreatM} by requiring only asymmetric conflicts. However, because of the asymmetric conflicts within this structure, it can be encoded much more easily than a \patternGreatM with symmetric conflicts.
This is also reflected by the fact that in the proofs for the separation result between \piMix and \piAsyn in \cite{petersNestmannGoltz13} we have to rely on the mutually exclusive nature of the conflicts in the \patternGreatM of the counterexample $ \piS{\piSE{1}, \ldots, \piSE{5}} $. Accordingly, we cannot use an \patternM or a \patternGreatM with asymmetric conflicts to derive separation results as done above. Instead, we show that, despite of the \patternGreatM with asymmetric conflicts, mobile ambients can be separated from \piMix by the synchronisation pattern \patternGreatM, because they cannot express a \patternGreatM with symmetric conflicts.

It turns out that the symmetric conflict in the pattern \patternM of the step $ b $ with $ a $ and $ c $ as given in Example~\ref{exa:nonLocalMMA} can only be expressed with an $ \maOpenAct $-action.

\begin{lemma}
	\label{lem:onesidedConflictM}
	Let $ P \in \procMA $ be an \patternM.
	Then one of the two conflicts is asymmetric or the step $ b $ reduces an $ \maOpenAct $-action.
\end{lemma}

\begin{proof}
	Assume the contrary, \ie assume there is some $ P \in \procMA $ such that $ P $ is an \patternM with two symmetric conflicts, where the step $ b $ reduces either an $ \maInAct $-action or an $ \maOutAct $-action.

	A step on an $ \maInAct $-action involves an $ \maIn{n} $-capability and two ambients: the ambient $ m $ that is moved by ``$ \maIn{n} $'', \ie in this, ``$ \maIn{n} $'' is located at top-level, and the ambient $ n $ to which $ m $ is moved.
	In this step, $ n $ is recurrent but $ m $ is non-recurrent.
	In order to obtain a symmetric conflict between $ b $ and $ a $ as well as $ c $, the later two steps have to consume the ``$ \maIn{n} $'' by providing an alternative $ n $ ambient on the same level in the ambient hierarchy (Case~1) or to move or open the ambient $ m $ (Case~2).
	\begin{compactitem}
		\item[Case~1:] Since $ a $ and $ c $ are parallel, at most one of them can consume ``$ \maIn{n} $''. Since this step moves the ambient $ m $, $ m $ is a non-recurrent capability of this step. Then the respective other step has to move or open $ m $ (Case~2). Hence, again, $ m $ is a non-recurrent capability of this step. As result in both of the steps $ a $ and $ c $, the same $ m $ is reduced as non-recurrent capability, which contradicts the assumption that $ a $ and $ c $ are parallel. Thus, neither $ a $ nor $ c $ can reduce ``$ \maIn{n} $''.
		\item[Case~2:] By Case~1, both of the steps $ a $ and $ c $ have to move or open the ambient $ m $. But then $ m $ is a non-recurrent capability of $ a $ and $ c $, which contradicts the assumption that $ a $ and $ c $ are parallel.
	\end{compactitem}
	Hence $ b $ cannot reduce an $ \maInAct $-action.

	A step on an $ \maOutAct $-action involves an $ \maOut{n} $-capability and two ambients: the ambient $ m $ that is moved by ``$ \maOut{n} $'', \ie in that ``$ \maOut{n} $'' is located at top-level, and the ambient $ n $ that surrounds $ m $ and out of which $ m $ is moved.
	In this step, $ n $ is recurrent, but $ m $ is non-recurrent.
	The ambient hierarchy forms a tree-structure. Thus, each ambient has at most one parent ambient.
	As a consequence, in order to obtain a symmetric conflict between $ b $ and $ a $ as well as $ c $, the later two steps have both to move or open the ambient $ m $. But then, $ m $ is a non-recurrent capability of $ a $ and $ c $, which contradicts the assumption that $ a $ and $ c $ are parallel. Hence, $ b $ cannot reduce an $ \maOutAct $-action.
\end{proof}

Since the synchronisation pattern \patternGreatM consists of several cyclic overlapping \patternM, all five steps of a \patternGreatM in mobile ambients have to reduce an $ \maOpenAct $-capability or at least one of the conflicts is asymmetric.
However, five steps on $ \maOpenAct $-capabilities cannot be combined in a cycle of odd degree. Thus, in all \patternGreatM-like structures there is at least one asymmetric conflict.
But there are no \patternGreatM (without asymmetric conflicts) in mobile ambients.

\begin{lemma}
	\label{lem:onesidedConflictGreatM}
	For all \patternGreatM-like structures $ P \in \procMA $ one of the conflicts in $ P $ that exist according to Definition~\ref{def:synchronisationPatternGreatM} is asymmetric.
\end{lemma}

\begin{proof}
	Assume the contrary, \ie assume $ P $ is a \patternGreatM such that the conflicts between $ a $ and $ b $, $ b $ and $ c $, $ c $ and $ d $, $ d $ and $ e $, and $ e $ and $ a $ are all symmetric conflicts.
	Then, by Lemma~\ref{lem:onesidedConflictM} and the Definitions~\ref{def:synchronisationPatternM} and \ref{def:synchronisationPatternGreatM}, all steps in $ \Set{ a, \ldots, c } $ reduce an $ \maOpenAct $-action. Hence each of these steps reduces exactly one $ \maOpenAct $-capability and exactly one ambient.
	Accordingly, two neighbouring steps compete for either an $ \maOpenAct $-capability or the ambient that is opened.
	Without loss of generality assume that $ a $ and $ b $ compete for an $ \maOpenAct $-capability. Then, since $ a $ and $ c $ are parallel, $ b $ and $ c $ have to compete for an ambient. Again, since $ b $ and $ d $ are parallel, $ c $ and $ d $ have to compete for an $ \maOpenAct $-capability. Then $ d $ and $ e $ have to compete for an ambient. Then $ e $ and $ a $ have to compete for an $ \maOpenAct $-capability. But then $ a $ and $ b $ cannot compete for an $ \maOpenAct $-capability, because $ e $ and $ b $ are parallel. The other case is similar.
\end{proof}

A \patternGreatM with an asymmetric conflict cannot be extended to a \patternGreatM that can be used as counterexample similarly to $ \piS{\piSE{1}, \ldots, \piSE{5}} $ in \cite{petersNestmannGoltz13, petersNestmannGoltz13TR}.
The proof to separate \piMix from \piSep and \piAsyn in \cite{petersNestmannGoltz13, petersNestmannGoltz13TR} exploits the fact that every maximal execution of \patternGreatM contains exactly two distributable steps of the five alternative steps that form the \patternGreatM.
But, if we replace a conflict in the \patternGreatM by an asymmetric conflict, then three steps are possible in one execution.

\begin{lemma}
	\label{lem:threeSteps}
	All \patternGreatM-like structures $ P \in \procMA $ have an execution that executes three of the five alternative steps that exist according to Definition~\ref{def:synchronisationPatternGreatM}.
\end{lemma}

\begin{proof}
	By Lemma~\ref{lem:onesidedConflictGreatM}, all \patternGreatM in mobile ambients have an asymmetric conflict. Because of that, whenever some $ \piS{\hole_a, \ldots, \hole_e} : \procMA^5 \to \procMA $ is such that for all $ \piSE{1}, \ldots, \piSE{5} \in \Set{ \maNull, \success } $ the term $ \piS{\piSE{1}, \ldots, \piSE{5}} $ is a \patternGreatM except for asymmetric conflicts, then there is a maximal execution of $ \piS{\piSE{1}, \ldots, \piSE{5}} $ that contains three steps of the set $ \Set{ a, \ldots, e } $: the two steps that are related by the asymmetric conflict (executing first the step that is not in conflict to the other and then the one-sided conflicting step) and the step that is in parallel to both of the former neighbouring steps.
\end{proof}

To show that there is no good and distributability-preserving encoding from \piMix into \MA we proceed as in \cite{petersNestmannGoltz13, petersNestmannGoltz13TR}.
First, we observe that every conflict in our counterexample $ \piS{\piSE{1}, \ldots, \piSE{5}} $ has to be translated into conflicts of the respective emulations in mobile ambients.

\begin{lemma}
	\label{lem:translateConflicts}
	Any good and distributability-preserving encoding $ \arbitraryEncoding : \procMix \to \procMA $ has to translate the conflicts in $ \piS{\piSE{1}, \ldots, \piSE{5}} $ into conflicts of the corresponding emulations.
\end{lemma}

The proof of this Lemma is exactly the same as the proof for the corresponding Lemma for encodings from \piMix into \piSep in \cite{petersNestmannGoltz13TR} but using the lemmas above, because this proof relies on the encodability criteria and the abstract notion of conflicts that is the same for \piSep and \MA.
Note that this proof assumes an encoding that satisfies compositionality as defined in Definition~\ref{def:compositionality}, but, as already stated in \cite{petersNestmannGoltz13}, it also holds in case of the relaxed version of compositionality that is used here.
Then, similar to Lemma~\ref{lem:splitUpB}, we show that each good encoding of the counterexample requires that a conflict has to be distributed.

\begin{lemma}
	\label{lem:splitGreatMAM}
	Any good and distributability-preserving encoding $ \arbitraryEncoding : \procMix \to \procMA $ has to split up at least one of the conflicts in $ \piS{\piSE{1}, \ldots, \piSE{5}} $ (or in $ \piS{\hole_a, \ldots, \hole_e} $) such that there exists a maximal execution of $ \ArbitraryEncoding{\piS{\hole_a, \ldots, \hole_e}} $ that emulates only one source term step, \ie unguards exactly one of the five holes.
\end{lemma}

\begin{proof}
	By operational completeness (Definition~\ref{def:operationalCorrespondence}), all five steps of $ \piS{ \hole_a, \ldots, \hole_e } $ have to be emulated in $ \ArbitraryEncoding{\piS{ \hole_a, \ldots, \hole_e }} $, \ie there are $ T_a\!\left( \hole_a, \ldots, \hole_e \right), \ldots, T_e\!\left( \hole_a, \ldots, \hole_e \right) \in \procMA^5 \to \procMA $ such that $ X_{\hole_a, \ldots, \hole_e}: \ArbitraryEncoding{\piS{ \hole_a, \ldots, \hole_e }} \steps T_x\!\left( \hole_a, \ldots, \hole_e \right) $ and $ T_x\!\left( \hole_a, \ldots, \hole_e \right) \asymp \ArbitraryEncoding{\hole_x} $ for all $ x \in \Set{ a, \ldots, e } $, where $ X $ is the upper case variant of $ x $.
	Because $ \piS{ \hole_a, \ldots, \hole_e } $ has no infinite execution and $ \arbitraryEncoding $ reflects divergence, $ \ArbitraryEncoding{\piS{ \hole_a, \ldots, \hole_e }} $ has no infinite execution.
	By compositionality (even in its relaxed form in Definition~\ref{def:goodEncoding}) and because $ \Names{\nullTerm} = \Names{\success} = \emptyset $, there exists a context $ \Context{}{}{\hole_a, \ldots, \hole_e} : \procMA^5 \to \procMA $ such that $ \ArbitraryEncoding{\piS{\hole_a, \ldots, \hole_e}} = \Context{}{}{\hole_a, \ldots, \hole_e} $ and $ \ArbitraryEncoding{\piS{\piSE{1}, \ldots, \piSE{5}}} = \Context{}{}{\ArbitraryEncoding{\piSE{1}}, \ldots, \ArbitraryEncoding{\piSE{5}}} $ for all $ \piSE{1}, \ldots, \piSE{5} \in \Set{ \nullTerm, \success } $.

	By Lemma~\ref{lem:translateConflicts}, the conflicts of $ \piS{ \hole_a, \ldots, \hole_e } $ have to be translated into conflicts of the respective emulations.
	Then, the encoding of $ \piS{ \hole_a, \ldots, \hole_c } $---if it exists---contains at least five steps, one for each of the emulations $ A_{\hole_a, \ldots, \hole_e}, \ldots, E_{\hole_a, \ldots, \hole_e} $, that capture the conflicts between the emulations of neighbouring source term steps.
	If one of the respective conflicts between these emulations is asymmetric there is a maximal execution of $ \ArbitraryEncoding{\piS{ \hole_a, \ldots, \hole_e }} $ that emulates two neighbouring source term steps, \ie unguards two holes of neighbouring steps.
	Let $ x, y \in \Set{ a, \ldots, e } $ such that the asymmetric conflict is between $ X $ and $ Y $ and such that $ \hole_x $ is the hole that can be unguarded first.
	Then there is $ T_{xy}\!\left( \hole_a, \ldots, \hole_e \right) \in \procMA^5 \to \procMA $ and an execution $ XY: \ArbitraryEncoding{\piS{ \hole_a, \ldots, \hole_c }} \steps T_x\!\left( \hole_a, \ldots, \hole_e \right) \steps T_{xy}\!\left( \hole_a, \ldots, \hole_e \right) $ that first emulates the source term step $ x $ and then unguards the hole $ \hole_y $, \ie in that $ \hole_y $ is unguarded in $ T_{xy}\!\left( \hole_a, \ldots, \hole_e \right) $.

	Let $ \piSE{y} = \success $ and $ \piSE{i} = \nullTerm $ for all $ i \in \left( \Set{ a, \ldots, e } \setminus \Set{ y } \right) $.
	Because of $ XY $, we have $ \ArbitraryEncoding{\piS{ \piSE{1}, \ldots, \piSE{5} }} \steps T_x\!\left( \piSE{1}, \ldots, \piSE{5} \right) \steps T_{xy}\!\left( \piSE{1}, \ldots, \piSE{5} \right) $.
	Then $ \piSE{x}\notReachSuccess $ and $ T_{xy}\!\left( \piSE{1}, \ldots, \piSE{5} \right)\reachSuccess $.
	By success sensitiveness, then $ \ArbitraryEncoding{\piSE{x}}\notReachSuccess $.
	Moreover, because of $ T_x\!\left( \piSE{1}, \ldots, \piSE{5} \right) \steps T_{xy}\!\left( \piSE{1}, \ldots, \piSE{5} \right) $, the reachability of success in the term $ T_{xy}\!\left( \piSE{1}, \ldots, \piSE{5} \right) $ implies $ T_x\!\left( \piSE{1}, \ldots, \piSE{5} \right)\reachSuccess $.
	Thus $ \ArbitraryEncoding{\piSE{x}} \asymp T_x\!\left( \piSE{1}, \ldots, \piSE{5} \right) $ violates the requirement that $ \asymp $ respects success.
	Hence, for each pair of neighbouring steps we need a symmetric conflict.

	Since $ \arbitraryEncoding $ preserves distributability (Definition~\ref{def:distributabilityPreservation}) and by Lemma~\ref{def:distributabilityPreservation}, each pair of distributable steps in $ \piS{ \hole_a, \ldots, \hole_e } $ has to be translated into emulations that are distributable in $ \ArbitraryEncoding{\piS{ \hole_a, \ldots, \hole_e }} $. Let $ X, Y, Z \in \Set{ A, B, C, D, E } $ be such that $ X_{\hole_a, \ldots, \hole_e} $ and $ Z_{\hole_a, \ldots, \hole_e} $ are distributable in $ \ArbitraryEncoding{\piS{ \hole_a, \ldots, \hole_e }} $ but $ Y_{\hole_a, \ldots, \hole_e} $ is in conflict with $ X_{\hole_a, \ldots, \hole_e} $ as well as $ Z_{\hole_a, \ldots, \hole_e} $.
	This implies that $ \ArbitraryEncoding{\piS{ \hole_a, \ldots, \hole_e }} $ is distributable into $ T_1, T_2 \in \procMA $ such that $ X_{\hole_a, \ldots, \hole_e} $ is an execution of $ T_1 $ and $ Z_{\hole_a, \ldots, \hole_e} $ is an execution of $ T_2 $. Since $ Y_{\hole_a, \ldots, \hole_e} $ is in conflict with $ X_{\hole_a, \ldots, \hole_e} $ and $ Z_{\hole_a, \ldots, \hole_e} $ and because all three emulations are executions of $ \ArbitraryEncoding{\piS{ \hole_a, \ldots, \hole_e }} $, there is one step of $ Y_{\hole_a, \ldots, \hole_e} $ that is in conflict with one step of $ X_{\hole_a, \ldots, \hole_e} $ and there is one (possibly the same) step of $ Y_{\hole_a, \ldots, \hole_e} $ that is in conflict with one step of $ Z_{\hole_a, \ldots, \hole_e} $. Moreover, since $ X_{\hole_a, \ldots, \hole_e} $ and $ Z_{\hole_a, \ldots, \hole_e} $ are distributable, if a single step of $ Y_{\hole_a, \ldots, \hole_e} $ is in conflict with $ X_{\hole_a, \ldots, \hole_e} $ as well as $ Z_{\hole_a, \ldots, \hole_e} $ then, by Lemma~\ref{lem:onesidedConflictM} and since asymmetric conflicts are not sufficient, this step reduces and $ \maOpenAct $-act in $ T_1 $ or $ T_2 $ and removes the border of an ambient in the respective other term.

	Assume that for all such combinations $ X $, $ Y $, and $ Z $, the conflicts between $ Y_{\hole_a, \ldots, \hole_e} $ and $ X_{\hole_a, \ldots, \hole_e} $ or $ Z_{\hole_a, \ldots, \hole_e} $ are ruled out by a single step of $ Y_{\hole_a, \ldots, \hole_e} $.
	Then this step reduces an $ \maOpenAct $-action in one of the executions $ X_{\hole_a, \ldots, \hole_e} $ and $ Z_{\hole_a, \ldots, \hole_e} $ and an ambient in the respective other, \ie $ X_{\hole_a, \ldots, \hole_e} $ and $ Y_{\hole_a, \ldots, \hole_e} $ compete either for an $ \maOpenAct $-action or for an ambient and $ Y_{\hole_a, \ldots, \hole_e} $ and $ Z_{\hole_a, \ldots, \hole_e} $ compete for the respective other kind. Without loss of generality let us assume that $ A_{\hole_a, \ldots, \hole_e} $ and $ B_{\hole_a, \ldots, \hole_e} $ compete for an $ \maOpen{n} $-action and, thus, $ B_{\hole_a, \ldots, \hole_e} $ and $ C_{\hole_a, \ldots, \hole_e} $ compete for the ambient $ n $, $ C_{\hole_a, \ldots, \hole_e} $ and $ D_{\hole_a, \ldots, \hole_e} $ compete for an $ \maOpen{n} $-action, $ D_{\hole_a, \ldots, \hole_e} $ and $ E_{\hole_a, \ldots, \hole_e} $ compete for another ambient $ n $, $ E_{\hole_a, \ldots, \hole_e} $ and $ A_{\hole_a, \ldots, \hole_e} $ compete for an $ \maOpen{n} $-action, and $ A_{\hole_a, \ldots, \hole_e} $ and $ B_{\hole_a, \ldots, \hole_e} $ compete for a third ambient $ n $. This is a contradiction, because $ A_{\hole_a, \ldots, \hole_e} $ and $ B_{\hole_a, \ldots, \hole_e} $ cannot compete for both an $ \maOpenAct $-action and an ambient.

	We conclude that there is at least one triple of emulations $ X_{\hole_a, \ldots, \hole_e} $, $ Y_{\hole_a, \ldots, \hole_e} $, and $ Z_{\hole_a, \ldots, \hole_e} $ such that the conflict of $ Y_{\hole_a, \ldots, \hole_e} $ with $ X_{\hole_a, \ldots, \hole_e} $ and with $ Z_{\hole_a, \ldots, \hole_e} $ results from two different steps in $ Y_{\hole_a, \ldots, \hole_e} $. Because $ X_{\hole_a, \ldots, \hole_e} $ and $ Z_{\hole_a, \ldots, \hole_e} $ are distributable, the reduction steps of $ X_{\hole_a, \ldots, \hole_e} $ that lead to the conflicting step with $ Y_{\hole_a, \ldots, \hole_e} $ and the reduction steps of $ Z_{\hole_a, \ldots, \hole_e} $ that lead to the conflicting step with $ Y_{\hole_a, \ldots, \hole_e} $ are distributable.
	We conclude, that there is at least one emulation of $ y $, \ie one execution $ Y_{\hole_a, \ldots, \hole_e}: \ArbitraryEncoding{\piS{ \hole_a, \ldots, \hole_e }} \steps T_y \asymp \ArbitraryEncoding{\hole_y} $, starting with two distributable executions such that one is (in its last step) in conflict with the emulation of $ x $ in $ X_{\hole_a, \ldots, \hole_e}: \ArbitraryEncoding{\piS{ \hole_a, \ldots, \hole_e }} \steps T_x \asymp \ArbitraryEncoding{\hole_x} $ and the other one is in conflict with the emulation of $ z $ in $ Z_{\hole_a, \ldots, \hole_e}: \ArbitraryEncoding{\piS{ \hole_a, \ldots, \hole_e }} \steps T_z \asymp \ArbitraryEncoding{\hole_z} $. In particular this means that also the two steps of $ Y_{\hole_a, \ldots, \hole_e} $ that are in conflict with a step in $ X_{\hole_a, \ldots, \hole_e} $ and a step in $ Z_{\hole_a, \ldots, \hole_e} $ are distributable.
	Hence, there is no possibility to ensure that these two conflicts are decided consistently, \ie there is a maximal execution of $ \ArbitraryEncoding{\piS{ \hole_a, \ldots, \hole_e }} $ that emulates $ X_{\hole_a, \ldots, \hole_e} $ but neither $ Y_{\hole_a, \ldots, \hole_e} $ nor $ Z_{\hole_a, \ldots, \hole_e} $.

	In the set $ \Set{ X_{\hole_a, \ldots, \hole_e} \mid X \in \Set{ A, B, C, D, E } } $ there are---apart from $ X_{\hole_a, \ldots, \hole_e} $, $ Y_{\hole_a, \ldots, \hole_e} $, and $ Z_{\hole_a, \ldots, \hole_e} $---two remaining executions.
	One of them, say $ X_{\hole_a, \ldots, \hole_e}' $, is in conflict with $ X_{\hole_a, \ldots, \hole_e} $ and the other one, say $ Z_{\hole_a, \ldots, \hole_e}' $, is in conflict with $ Z_{\hole_a, \ldots, \hole_e} $.
	Since $ X_{\hole_a, \ldots, \hole_e} $ is emulated successfully, $ X_{\hole_a, \ldots, \hole_e}' $ cannot be emulated.
	Moreover, note that $ Y_{\hole_a, \ldots, \hole_e} $ and $ Z_{\hole_a, \ldots, \hole_e}' $ are distributable.
	Thus, also $ Z_{\hole_a, \ldots, \hole_e}' $ and the partial execution of $ Y_{\hole_a, \ldots, \hole_e} $ that leads to the conflict with $ Z_{\hole_a, \ldots, \hole_e} $ are distributable.
	Moreover, also the step of $ Y_{\hole_a, \ldots, \hole_e} $ that already rules out $ Z_{\hole_a, \ldots, \hole_e} $ cannot be in conflict with a step of $ Z_{\hole_a, \ldots, \hole_e}' $.
	Thus, although the successful completion of $ Z_{\hole_a, \ldots, \hole_e} $ is already ruled out by the conflict with $ Y_{\hole_a, \ldots, \hole_e} $, there is some step of $ Z_{\hole_a, \ldots, \hole_e} $ left, that is in conflict with one step in $ Z_{\hole_a, \ldots, \hole_e}' $.
	Hence, the conflict between $ Z_{\hole_a, \ldots, \hole_e} $ and $ Z_{\hole_a, \ldots, \hole_e}' $ cannot be ruled out by the partial execution described so fare that leads to the emulation of $ X_{\hole_a, \ldots, \hole_e} $ but forbids to complete the emulations of $ X_{\hole_a, \ldots, \hole_e}' $, $ Y_{\hole_a, \ldots, \hole_e} $, and $ Z_{\hole_a, \ldots, \hole_e} $.
	Thus, it cannot be avoided that $ Z_{\hole_a, \ldots, \hole_e} $ wins this conflict, \ie that also $ Z_{\hole_a, \ldots, \hole_e}' $ cannot be completed.
	We conclude that there is a maximal execution of $ \ArbitraryEncoding{\piS{ \hole_a, \ldots, \hole_e }} $ such that only one of the five source term steps of $ \piS{ \hole_a, \ldots, \hole_e } $ is emulated and only one hole is unguarded.
\end{proof}

Again, the above proof is in its main idea similar to the respective proof of the corresponding result for encodings from \piMix into \piSep in \cite{petersNestmannGoltz13TR}.
However, since that proof depends on the expressive power of the considered target language to reason about the properties of the counterexample, we have to adapt it to mobile ambients.
Finally, we show again that this distribution of the conflict rules out the possibility of a good and distributability-preserving encoding.

\begin{theorem}
	\label{thm:noGoodEncodingMixSepGreatM}
	There is no good and distributability-preserving encoding from \piMix into \MA.
\end{theorem}

The proof of this Theorem very closely follows the proof of the corresponding Theorem for encodings from \piMix into \piSep in \cite{petersNestmannGoltz13TR}.
It picks the maximal execution of the translation that unguards|according to Lemma~\ref{lem:splitGreatMAM}|only one hole $ \hole_x $ by emulating only one step $ x $ of $ \piS{\piSE{1}, \ldots, \piSE{5}} $.
Then, we can choose $ \piSE{1}, \ldots, \piSE{5} \in \Set{ \maNull, \success } $ such that $ \piSE{x} = \nullTerm = \piSE{y} $, where $ y $ is one of the two steps that is parallel to $ x $, and $ \piSE{z} = \success $ for all other cases.
Accordingly, for the result $ S_x $ of the step $ x: \piS{\piSE{1}, \ldots, \piSE{5}} \step S_x $, we have $ S_x\notMustReachSuccessFinite $, by doing $ y $ next, but $ S_x\reachSuccess $, because of success in the respective other step that can be executed after $ x $.
However, the maximal execution of $ \piS{\hole_a, \ldots, \hole_e} $ that unguards only $ \hole_x $ and emulates only $ x $ cannot have the same behaviour \wrt success.
After emulating $ x $ we reach a term that cannot offer the possibility to reach success (without the emulation of another source term step) as well as to deadlock without reaching success.
This violates our requirements on good encodings.

\section{Distributing Mobile Ambients}
\label{sec:distributeMA}

Theorem~\ref{thm:noGoodEncodingMA} shows that mobile ambients are not as distributable as the \joinCal.
Nonetheless, \cite{fournetLevySchmitt00} presents an encoding from \MA into \join in order to build a distributed implementation of mobile ambients in Jocaml (\cite{jocaml99}).
Let us consider what this encoding does with our counterexample $ \maM $ for the non-existence of a good and distributability-preserving encoding from \MA into \join.
The encoding in \cite{fournetLevySchmitt00} translates each ambient into a single unique join definition.
Then it splits $ \maInAct $, $ \maOutAct $, and $ \maOpenAct $-actions into respective subactions that are controlled by the join definition that represents the parent ambient in the source.
Therefore, to perform the emulations of the distributed steps $ a $ and $ c $ of $ \maM $, the respective parts of the implementation first have to register their desire to do these steps with their parent join definition.
Unfortunately, as each join definition is a single location, these two steps interact with the same join definition, so they cannot be considered as distributed.
Accordingly, the encoding presented in \cite{fournetLevySchmitt00} is not distributability-preserving in our sense, because the emulations of $ a $ and $ c $ are synchronised.

Indeed, the authors of \cite{fournetLevySchmitt00} already state that the explicit control of subactions by the translation of the parent ambient introduces some form of synchronisation.
However, they claim that the form of synchronisation introduced by the presented encoding is less crucial than, \eg, a centralised solution.
Our results support the quality of their solution, by proving that no good and fully distributability-preserving encoding from \MA into \join exists.
So, a bit of synchronisation is indeed necessary.
But, our results also suggest possible ways to circumvent the problems in the distribution of mobile ambients altogether by proposing small alterations of the source calculus itself in order to prevent \patternM-patterns from the outset.

By Lemma~\ref{lem:onesidedConflictM}, all \patternM in mobile ambients rely on a conflict with an $ \maOpenAct $-action that addresses two different ambients with the same name.
A natural solution to circumvent this problem is to avoid different ambients with the same name.
By Lemma~\ref{lem:onesidedConflictM}, mobile ambients with unique ambient names cannot express the pattern \patternM.

\begin{corollary}
	There are no \patternM in mobile ambients, where all ambient names are unique.
\end{corollary}

Without such an \patternM as counterexample, our proof of Theorem~\ref{thm:noGoodEncodingMA} would no longer work.
Instead, we can show that there is then no good and distributability-preserving encoding from \piAsyn into \MA, by using the example of an \patternM in \piAsyn of \cite{petersNestmannGoltz13} as counterexample and following a similar proof strategy as for the separation result between \piAsyn and \join.

\begin{claim}
  If mobile ambients forbid for ambients with the same name, then there is no good and distributability-preserving encoding from \piAsyn into \MA.
\end{claim}

The proof of the above claim relies of the formalisation of the requirement that no two different ambients have the same name in the definition of the calculus.
More precisely, we need to adapt the proof that every good and distributability-perserving encoding has to split up the conflict in the \patternM of $ b $ with $ a $ and $ c $ to the target language \MA with unique ambient names.
Since there are several different ways to implement this requirement in the syntax of mobile ambients, we do not formally prove the above claim here.
However, we expect that this proof would exploit the same strategy as in \cite{petersNestmannGoltz13TR} and require only small adaptations due to the definition of the calculus.

Actually, the possibility to have different ambients with the same name was already identified as problematic in the encoding of \cite{fournetLevySchmitt00}.
To circumvent this problem, the encoding introduces unique identifiers for all ambients and one of the reasons for the interaction with the respective translation of the parent ambient to control the translations of ambient actions is that these translations of parent ambients keep the knowledge about the unique identifiers of their children.
Thus, forbidding different ambients with the same name not only allows for completely distributed implementations of the calculus but also significantly simplifies translations that follow the strategy of \cite{fournetLevySchmitt00}.

To obtain strategies to implement this requirement, we can have a look at other distributed calculi with unique location names.
The \joinCal (\cite{fournetGonthier96}) ensures the uniqueness of its locations by combining input prefixes with restriction in join definitions.
Thus, every join definition, \ie location, introduces its own name space.
Interaction is limited to such restricted names with a clear and unique destination.
The advantage is that the uniqueness of location names is ensured by definition;
the disadvantage is that some forms of interaction---\eg a two-way handshake---are syntactically more difficult due to these sharp restriction borders.
The distributed \piCal (\cite{hennessy07}) has a flat structure of locations and ensures uniqueness by the structural congruence rule $ \maLoc{n}{P} \mid \maLoc{n}{Q} \equiv \maLoc{n}{P \mid Q} $ that unifies different parts of a location.
However, adding such a rule to mobile ambients requires a non-trivial adaptation of the semantics, because the $ \maOpenAct $, $ \maInAct $, and $ \maOutAct $-actions would need to first collect all ambient parts that are possibly dispersed over the term structure before they can proceed.
Moreover, following this approach would not completely rule out different ambients with the same name but only different such ambients in the same parent ambient (or at top-level). This is, however, sufficient to ensure that there are no \patternM.

\section{Conclusions}
\label{sec:conclusions}

We proved that there is no good and distributability-preserving encoding from mobile ambients (\MA) into the \joinCal (\join) and neither from the standard \piCal with mixed choice (\piMix) into mobile ambients.
Note that these results stay valid also for the extension of \MA with communication prefixes as described in \cite{cardelliGordon98, cardelliGordon00}, because these communications are local steps that cannot be in conflict to steps with $ \maInAct $, $ \maOutAct $, or $ \maOpenAct $-actions.
Thus, all conflicts added by the extension with communication primitives are local and not relevant for the preservation of distributability.
Consequently, by extending the results of \cite{petersNestmannGoltz13}, we place mobile ambients on the same level as the \piCal with separate choice (\piSep) and the asynchronous \piCal (\piAsyn) above \join and below \piMix.
As visualized in Figure~\ref{fig:hierarchy}, mobile ambients contain non-local \patternM but cannot express a non-local \patternGreatM without asymmetric conflicts.

\begin{figure}[t]
	\centering
	\begin{tikzpicture}[node distance=3cm, auto]
		\node (mix)		at (0,1.8)		{\piMix};
		\node (sep)		at (-1.5, 0)	{\piSep};
		\node (asyn)	at (0, 0)		{\piAsyn};
		\node (ma)		at (1.5, 0)		{\MA};
		\node (join)	at (0, -1.8)	{\join};

		\draw[thick] node[star, star points=5,star point ratio=2.25, draw=black, inner sep=3.5pt] at (3, 1) {};
		\draw[dashed] (-2, 1) -- (2,1) -- (3,2);
		\draw[dashed] (2,1) -- (3,0);

		\draw[thick] (-3.5,-1.2) -- (-3.25, -.7) -- (-3, -1.2) -- (-2.75,-.7) -- (-2.5,-1.2);
		\draw[dashed] (2,-1) -- (-2,-1) -- (-3,-2);
		\draw[dashed] (-2,-1) -- (-3,0);
	\end{tikzpicture}
	\vspace*{-1em}
	\caption{Distributability in Pi-like Calculi.}
	\label{fig:hierarchy}
\end{figure}
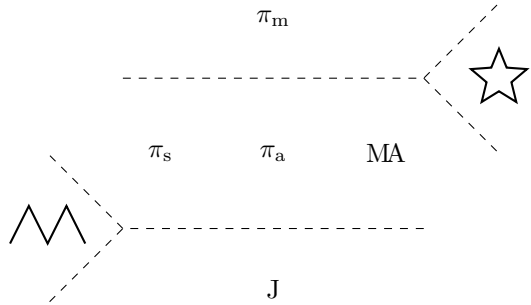

Asymmetric conflicts, as present in mobile ambients, constitute a variant of conflicts that turns out to be not as crucial for distributed implementations as the standard symmetric conflicts that we usually find in calculi.  Nonetheless, the existence of non-local \patternM make fully distributed implementations of mobile ambients difficult---as already observed in \cite{fournetLevySchmitt00}.  However, since the reason for these difficulties is now clearly captured in a simple synchronisation pattern, we can more easily derive strategies to adapt mobile ambients to a distributed calculus without such problems.

Interestingly, the extension of mobile ambients into mobile safe ambients in \cite{leviSangiorgi03} does not solve this problem.
The main idea of safe ambients is that actions require an explicit agreement on this action by both participating ambients. Therefore, safe ambients augment the respective target ambient of an action~$ a $ with a matching complementary action $ \overline{a} $.
This extension, however does neither change the power to express the pattern \patternM nor the asymmetric nature of conflicts with steps that do not rely on an $ \maOpenAct $-action.
In fact, the $ \maM $ in mobile ambients, \ie the pattern \patternM, becomes
\begin{center}
	$ \left( \maOpen{n_1} \mid \maLoc{n_1}{\maOpenO{n_1} \mid P_1} \right) \mid \left( \maLoc{n_1}{\maOpenO{n_1} \mid \maIn{n_2}.P_2} \mid \maLoc{n_2}{\maInO{n_1} \mid P_3} \right) $
\end{center}
in safe ambients.
This term is again an \patternM sharing the kind of steps and properties of $ \maM $.
Thus, we obtain the same separation result as in Theorem~\ref{thm:noGoodEncodingMA} with safe ambients using the above counterexample.
Moreover, since safe ambients do also not contain \patternGreatM, also Theorem~\ref{thm:noGoodEncodingMixSepGreatM} stays valid for safe ambients.

The most obvious way to obtain a fully distributed variant of mobile ambients is to ensure uniqueness of ambient names.
As a consequence, actions of mobile ambients have a clear and unique destination.
Note that, having clear and unique destinations for all actions that travel location borders is also crucial for the distributability of other calculi such as the \joinCal or the distributed \piCal.
Such unique destinations significantly limit the possibility of conflicts and ensure that all remaining conflicts of the language are local.
As a consequence, distributed implementations of such languages do not need to introduce synchronisations and, thus, do not change their semantics.
Hence, keeping the destinations for all actions that travel location borders unique, is a good strategy to build distributed calculi in general.


\bibliography{DistributabilityMobileAmbients}

\end{document}